\documentclass[a4paper,UKenglish,cleveref, autoref]{lipics-v2019}
\title{An Improved FPTAS for 0-1 Knapsack}

\author{Ce Jin}{Institute for Interdisciplinary Information Sciences, Tsinghua University, Beijing, China}{jinc16@mails.tsinghua.edu.cn}{}{}

\authorrunning{C.\,Jin}

\Copyright{Ce Jin}%mandatory, please use full first names. LIPIcs license is "CC-BY";  http://creativecommons.org/licenses/by/3.0/


%mandatory

\relatedversion{}%optional, e.g. full version hosted on arXiv, HAL, or other respository/website
\supplement{}%optional, e.g. related research data, source code, ... hosted on a repository like zenodo, figshare, GitHub, ...

\acknowledgements{
Part of this research was done while visiting Harvard University. I would like to thank Professor Jelani Nelson for introducing this problem to me, advising this project, and giving many helpful comments on my writeup. 
}%optional

\nolinenumbers %uncomment to disable line numbering

\hideLIPIcs  %uncomment to remove references to LIPIcs series (logo, DOI, ...), e.g. when preparing a pre-final version to be uploaded to arXiv or another public repository

\EventEditors{Christel Baier, Ioannis Chatzigiannakis, Paola Flocchini, and Stefano Leonardi}
\EventNoEds{4}
\EventLongTitle{46th International Colloquium on Automata, Languages, and Programming (ICALP 2019)}
\EventShortTitle{ICALP 2019}
\EventAcronym{ICALP}
\EventYear{2019}
\EventDate{July 9--12, 2019}
\EventLocation{Patras, Greece}
\EventLogo{eatcs}
\SeriesVolume{132}
\ArticleNo{71}
\begin{document}
\theoremstyle{plain}
\newtheorem{prop}[theorem]{Proposition}
\maketitle

\begin{abstract}
The 0-1 knapsack problem is an important NP-hard problem that admits fully polynomial-time approximation schemes (FPTASs). Previously the fastest FPTAS by Chan (2018) with approximation factor $1+\varepsilon$ runs in $\tilde O(n + (1/\varepsilon)^{12/5})$ time, where $\tilde O$ hides polylogarithmic factors.
In this paper we present an improved algorithm in $\tilde O(n+(1/\varepsilon)^{9/4})$ time, with only a $(1/\varepsilon)^{1/4}$ gap from the quadratic conditional lower bound based on $(\min,+)$-convolution.  Our improvement comes from a multi-level extension of Chan's number-theoretic construction, and a greedy lemma that reduces unnecessary computation spent on cheap items.
\end{abstract}

\newcommand{\R}{\mathbb{R}}
\newcommand{\Z}{\mathbb{Z}}
\newcommand{\eps}{\varepsilon}
\renewcommand{\P}{\mathbb{P}}

\begin{comment}
The presentation of the paper would benefit from adding more intuition for the lemma statements and proof steps, and from adding more comparisons of the lemma statements with the corresponding lemmas of Chan's.

final: grammar check
final: number of pages
\end{comment}

\section{Introduction}
\subsection{Background}
In the \textit{0-1 knapsack} problem, we are given a set $I$ of $n$ items where each item $i \in I$ has weight $w_i$ and profit $p_i$, and we want to select a subset $J\subseteq I$ such that $\sum_{j\in J}w_j \le W$ and $\sum_{j\in J}p_j$ is maximized.

The 0-1 knapsack problem is a fundamental optimization problem in computer science and is one of Karp's 21 NP-complete problems \cite{karp1972reducibility}.
An important field of study on NP-hard problems is to find efficient approximation algorithms.
A $(1+\eps)$-approximation algorithm (for a maximization problem) outputs a value $\mathrm{SOL}$ such that $\mathrm{SOL}\le \mathrm{OPT} \le (1+\eps)\cdot \mathrm{SOL}$, where $\mathrm{OPT}$ denotes the optimal answer. The 0-1 knapsack problem is one of the first problems  that were shown to have fully polynomial-time approximation schemes (FPTASs), i.e., algorithms with approximation factor $1+\eps$ for any given $0<\eps<1$ and running time polynomial in both $n$ and $1/\eps$. 

\begin{table}
\centering
\caption{FPTASs for 0-1 knapsack}
\label{history}
\begin{tabular}{ l|l|l } 
\hline \hline
 $O( n \log n + ( \frac{1}{\eps} )^4 \log \frac{1 }{\eps})$& Ibarra and Kim \cite{ibarra1975fast}& 1975 \\ 
 $O(n\log\frac{1}{\eps} + (\frac{1}{\eps})^4)$& Lawler \cite{lawler1979fast}& 1979 \\ 
 $O(n\log\frac{1}{\eps} + (\frac{1}{\eps})^3 \log^2 \frac{1}{\eps})$& Kellerer and Pferschy \cite{kellerer2004improved}& 2004 \\ 
$O(n\log\frac{1}{\eps} + (\frac{1}{\eps})^{5/2} \log^3 \frac{1}{\eps})$  (randomized)& Rhee \cite{rhee2015faster}& 2015\\
$O(n \log \frac{1}{\eps} + (\frac{1}{\eps})^{12/5}/2^{\Omega(\sqrt{\log(1/\eps)})})$& Chan \cite{chan2018approximation} & 2018\\
$O(n \log \frac{1}{\eps} + (\frac{1}{\eps})^{9/4}/2^{\Omega(\sqrt{\log(1/\eps)})})$& \textbf{This work} & \\
 \hline
 $O(\frac{1}{\eps} n^3)$ & Textbook algorithm & \\ 
 $O(\frac{1}{\eps} n^2)$ & Lawler \cite{lawler1979fast} & 1979\\ 
  $O( (\frac{1}{\eps}  )^2  n \log\frac{1}{\eps})$ & Kellerer and Pferschy \cite{kellerer1999new} & 1999\\ 
  $\tilde O( \frac{1}{\eps}    n^{3/2} )$ (randomized, Las Vegas) & Chan \cite{chan2018approximation} & 2018\\ 
  $ O( ((\frac{1}{\eps})^{4/3}n + (\frac{1}{\eps})^2)/ 2^{\Omega(\sqrt{\log(1/\eps)})} )$  & Chan \cite{chan2018approximation} & 2018\\  
 $O(   ((\frac{1}{\eps})^{3/2} n^{3/4}+ (\frac{1}{\eps})^2 ) /2^{\Omega(\sqrt{\log(1/\eps)})}+n\log\frac{1}{\eps}) $   & \textbf{This work} & \\  
  \hline\hline
\end{tabular}
\end{table}
There has been a long line of research on finding faster FPTASs for the 0-1 knapsack problem, as summarized in Table \ref{history}. 
The first algorithm with only subcubic dependence on $1/\eps$ was due to Rhee \cite{rhee2015faster}.  
Very recently, Chan \cite{chan2018approximation} gave an elegant algorithm for the 0-1 knapsack problem in deterministic $O(n \log \frac{1}{\eps} + (\frac{1}{\eps})^{5/2}/2^{\Omega(\sqrt{\log(1/\eps)})})$ via simple combination of the SMAWK algorithm \cite{aggarwal1987geometric} and a standard divide-and-conquer technique. The speedup of superpolylogarithmic factor $2^{\Omega(\sqrt{\log(1/\eps)})}$ is due to recent progress on $(\min,+)$-convolution \cite{bremner2014necklaces,williams2014faster,chan2016deterministic}. Using an elementary number-theoretic lemma, Chan further improved the algorithm to $O(n \log \frac{1}{\eps} + (\frac{1}{\eps})^{12/5}/2^{\Omega(\sqrt{\log(1/\eps)})})$ time, and obtained two new algorithms running in $\tilde O(\frac{1}{\eps} n^{3/2})$ and $O((\frac{1}{\eps})^{4/3} n + (\frac{1}{\eps})^2)/2^{\Omega(\sqrt{\log(1/\eps)})})$ time respectively, which are faster for small $n$.

FPTASs on several special cases of 0-1 knapsack are also of interest.
For the \textit{unbounded knapsack} problem, where every item has infinitely many copies, Jansen and Kraft \cite{jansen2018faster} obtained an $O(n + ( \frac{1}{\eps} )^2 \log^3 \frac{1}{\eps} )$-time algorithm; the unbounded version can be reduced to 0-1 knapsack with only a logarithmic blowup in the problem size \cite{cygan2019problems}.  
For the \textit{subset sum} problem, where every item has $p_i=w_i$, Kellerer et al.  \cite{kellerer2003efficient} obtained an algorithm with $O(\min \{n/\eps, n + (\frac{1 }{\eps} )^2 \log\frac{1}{\eps}\})$ running time, which will be used in our algorithm as a subroutine. For the \textit{partition} problem, which is a special case of the subset sum problem where $W = \frac{1}{2}\sum_{i\in I}w_i$, Mucha et al. \cite{mucha2019subquadratic} obtained an algorithm with a subquadratic  $\tilde O(n+(\frac{1}{\eps})^{5/3})$ running time. 

On the lower bound side, recent reductions showed by Cygan~et~al. \cite{cygan2019problems} and  
K\"{u}nnemann et~al. \cite{kunnemann2017fine} imply  that 0-1 knapsack and unbounded knapsack have no FPTAS in $O((n+\frac{1}{\eps})^{2-\delta})$ time, unless $(\min,+)$-convolution has truly subquadratic algorithm \cite{mucha2019subquadratic}. It remains open whether 0-1 knapsack has a matching upper bound.   

\subsection{Our results}
In this paper we present improved FPTASs for the 0-1 knapsack problem. Our results are summarized in the following two theorems.
\begin{theorem}
\label{mainthm}
  There is a deterministic $(1+\eps)$-approximation algorithm for 0-1 knapsack with running time $O(n \log \frac{1}{\eps} + (\frac{1}{\eps})^{9/4} /2^{\Omega(\sqrt{\log(1/\eps)})})$.
\end{theorem}

\begin{theorem}
\label{mainsmall}
   For $n=O(\frac{1}{\eps})$, there is a deterministic $(1+\eps)$-approximation algorithm for 0-1 knapsack with running time $O\Big (  \big (n^{3/4}(\frac{1}{\eps})^{3/2} + (\frac{1}{\eps})^2\big ) /2^{\Omega(\sqrt{\log(1/\eps)})}\Big)$.
\end{theorem}

Theorem \ref{mainsmall} gives the current best time bound for $(\frac{1}{\eps})^{2/3}\ll n \ll \frac{1}{\eps}$, improving upon the previous $O((\frac{1}{\eps})^{4/3} n + (\frac{1}{\eps})^2)/2^{\Omega(\sqrt{\log(1/\eps)})})$ algorithm by Chan \cite{chan2018approximation}.
For $n\ll (\frac{1}{\eps})^{2/3}$, Chan's $\tilde O(\frac{1}{\eps}n^{3/2})$ time randomized algorithm \cite{chan2018approximation} remains the fastest.

For $n \gg \frac{1}{\eps}$, Theorem \ref{mainthm} gives a better time bound, improving upon the previous $O(n \log \frac{1}{\eps} + (\frac{1}{\eps})^{12/5} /2^{\Omega(\sqrt{\log(1/\eps)})})$ algorithm by Chan \cite{chan2018approximation}.

\subsection{Outline of our algorithm}
\label{secoutline}
We give an informal overview of our improved algorithm for 0-1 knapsack.

Using a known reduction \cite{chan2018approximation}, it suffices to solve an easier instance of 0-1 knapsack where  profits of all items satisfy $p_i\in [1,2]$. Here ``solving an instance'' means approximating the function $f(x) := [\text{maximum total profit of items with at most $x$ total weight}]$ for all $x\ge 0$, rather than for just a single point $x=W$.    
In this restricted case, simple greedy (sorting according to \textit{unit profits} $p_i/w_i$) gives an additive error of at most $\max_j p_j = O(1)$, so it suffices to approximate the capped function $\min \{\eps^{-1}, f(x)\}$ with approximation factor $1+O(\eps)$.
Chan gave an algorithm that $(1+\eps)$-approximates $\min\{B,f(x)\}$ in $\tilde O(n+\eps^{-2} B^{1/2})$ time (implied by \cite[Lemma 7]{chan2018approximation}), which immediately implies an $\tilde O(n+\eps^{-5/2})$ time FPTAS by setting $B=\eps^{-1}$.
\subparagraph*{Greedy.}
Now we explain how to use a greedy argument (described in detail in Section \ref{secmain}) to improve this algorithm to $\tilde O(n+\eps^{-7/3})$ time. 
We sort all items (with $p_i\in [1,2]$) in non-increasing order of unit profits $p_i/w_i$, and divide them into three subsets $H,M,L$ (items with high, medium, low unit profits), where $H$ contains the top $\Theta(\eps^{-1})$ items, and $L$ contains all items $i$ for which $p_i/w_i \le (1-\eps^{2/3})\cdot \min_{h\in H}\{p_h/w_h\} $, so there is a gap between the unit profits of $H$-items and $L$-items. Intuitively, there are sufficiently many $H$-items available, so it's not optimal to include too many cheap $L$-items when the knapsack capacity is not very big.  To be more precise,   we prove that in any optimal solution we care about (i.e., having optimal total profit smaller than $\eps^{-1}$), the total profit contributed by $L$-items cannot exceed $O(\eps^{-2/3})$.  Hence, for subset $L$ we only need to approximate up to $B=\Theta(\eps^{-2/3})$ in $\tilde O(n+\eps^{-2}B^{1/2}) = \tilde O(n+\eps^{-7/3})$ time.
Subset $H$ has only $O(\eps^{-1})$ items and can be solved using Chan's $\tilde O(\eps^{-4/3}n+\eps^{-2})$ algorithm in $\tilde O(\eps^{-7/3})$ time. To solve subset $M$, we round down the profit value $p_i$ for every item $i\in M$, so that the unit profit $p_i/w_i$ becomes a power of $(1+\eps)$. Then there are $O(\eps^{-1/3})$ distinct unit profit values in $M$. Items with the same unit profit can be solved together using the efficient FPTAS for \textit{subset sum} by Kellerer et al. \cite{kellerer2003efficient} in $\tilde O(n+\eps^{-2})$ time. Finally we merge the results for $H,M,L$. The total time complexity is $\tilde O(n+\eps^{-7/3})$. 
\subparagraph*{Multi-level number-theoretic construction.}
The above approach invokes two of Chan's algorithms: an $\tilde O(n+\eps^{-2}B^{1/2})$ algorithm (useful for small $B$) and an $\tilde O(\eps^{-4/3}n+\eps^{-2})$ algorithm (useful for small $n$).   
The key ingredient in these algorithms is a number-theoretic lemma: we can $(1+\eps)$-approximate all profit values $p_i\in [1,2]$ by multiples of elements from a small set $\Delta\subset[\delta,2\delta]$ of size $|\Delta|=\tilde O(\frac{\delta}{\eps})$  (small $|\Delta|$ can reduce the additive error incurred from rounding). 

Chan obtained an $\tilde O(n+\eps^{-2}B^{2/5})$ time algorithm using some additional tricks. First, evenly partition $\Delta$ into $r$ subsets $\Delta^{(1)},\dots,\Delta^{(r)}$, and divide the items into $P=P^{(1)}\cup \dots \cup P^{(r)}$ accordingly, so that profit values from $P^{(j)}$ are approximated by $\Delta^{(j)}$-multiples.  To $(1+\eps)$-approximate the profit function $f_j$ for each $P^{(j)}$, pick a threshold $B_0\ll B$, and return the combination of  a $(1+\eps)$-approximation of $\min\{f_j, B_0\}$ and  an $\eps B_0$-additive-approximation of  $\min\{f_j,B\}$. Since the size of $\Delta^{(j)}$ is only $|\Delta|/r$, the latter function can be approximated faster when $r\gg 1$. 
Finally, merge $f_j$ over all $1\le j\le r$. By fine-tuning the parameters $r,\delta,B_1$, the time complexity is improved to $\tilde O(n+\eps^{-2}B^{2/5})$.

Our new algorithm extends this technique to \textit{multiple levels}. To $(1+\eps)$-approximate the profit function $f_j$ for each $P^{(j)}$, we will pick $B_0 \ll B_1\ll \dots \ll B_{d-1} \ll B_d\approx B$, and compute the $\eps B_{i-1}$-additive-approximation of $\min\{f_j,B_{i}\}$, for all $i\in [d]$.
An issue of this multi-level approach is that, different levels have different optimal parameters $\delta_i$ and different $\Delta_i^{(1)},\dots,\Delta_i^{(r)}$, but we have to stick to the same partition of items $P=P^{(1)}\cup \dots \cup P^{(r)}$ over all levels. We overcome this issue by enforcing that $\Delta_i^{(j)}$ at level $i$ must be generated by multiples of elements from $\Delta_{i-1}^{(j)}$ at level $i-1$, so that $P^{(j)}$ can be approximated by $\Delta_i^{(j)}$-multiples for all levels.  To achieve this, we need a multi-level version of the number-theoretic lemma. We will discuss this part in detail in Section \ref{secextend}.

Using this multi-level construction, we obtain algorithms in $\tilde O(n+\eps^{-2}B^{1/3})$ time and $\tilde O(\eps^{-3/2} n^{3/4}+\eps^{-2})$ time. Combining these improved algorithms with the greedy argument previously described (the threshold which splits $M$ and $L$ needs to be adjusted accordingly), we obtain an algorithm in  $\tilde O(n+\eps^{-9/4})$ time as claimed in Theorem \ref{mainthm}.
 
\section{Preliminaries}
\label{secprelim}
Throughout this paper, $\log x$ stands for $\log_2 x$, and $\tilde O(f)$ stands for $O(f\cdot \text{poly} \log (f))$. 

We will describe our algorithm with approximation factor $1+O(\eps)$, which can be lowered to $1+\eps$ if  we scale down $\eps$ by a constant factor at the beginning. 

We are only interested in the case where $n =O(\eps^{-4})$.  For greater $n$,  Lawler's $O(n\log \frac{1}{\eps} + (\frac{1}{\eps})^4)$ algorithm \cite{lawler1979fast} is already near-optimal. Hence we assume $\log n = O(\log \eps^{-1})$.

Assume $0<w_i\le W$ and $p_i>0$ for every item $i$. Then a trivial lower bound of the maximum total profit is $\max_j p_j$. At the beginning, we discard all items $i$ with $p_i \le \frac{\eps}{n} \max_j p_j$. Since the total profit of discarded items is at most $\eps \max_j p_j$, the optimal total profit is only reduced by a factor of $1+O(\eps)$. So we can assume that $\frac{\max_j p_j}{\min_j p_j} \le \frac{n}{\eps}$.

We adopt Chan's terminology in presenting our algorithm.  For a set $I$ of items, define the profit function
\begin{equation*}
f_I(x) = \max \Bigg \{ \sum_{i\in J} p_i  : \sum_{i \in J} w_i\le x, \;\; J \subseteq I \Bigg\}
\end{equation*}
over non-negative real numbers $x\ge 0 $. Note that $f_I$ is a monotone (nondecreasing) step function. The \textit{complexity} of a monotone step function refers to the number of its steps. 
  
  We say that a function $\tilde f$ approximates a function $f$ with factor $1+\eps$ if $\tilde f(x)\le f(x) \le (1+\eps)\tilde f(x)$ for all $x\ge 0$.  We say that $\tilde f$ approximates $f$ with additive error $\delta$ if $\tilde f(x)\le f(x)  \le  \tilde f (x) + \delta$ for all $x \ge 0$. Our goal is to approximate $f_I$ with factor $1+O(\eps)$ on the input item set $I$.

Let $I_1,I_2$ be two disjoint subsets of items, and $I = I_1 \cup I_2$. We have $f_{I} = f_{I_1} \oplus f_{I_2}$, where $\oplus$ denotes the $(\max,+)$-convolution, defined by 
$(f\oplus g)(x) = \max_{0\le x'\le x}(f(x')+g(x-x'))$. If two non-negative monotone step functions $f,g$ are approximated with factor $1+\eps$ by functions $\tilde f,\tilde g$ respectively, then $f \oplus g$ is also approximated by $\tilde f \oplus \tilde g$ with factor $1+\eps$.

For a monotone step function $f$ with range\footnote{Here \textit{range} refers to the set of possible output values of the function.} contained in $\{0\} \cup [A,B]$, we can obtain a function $\tilde f$ with complexity only $O(\eps^{-1}\log (B/A))$ which approximates $f$ with factor $1+\eps$, by simply rounding $f$ down to powers of $(1+\eps)$.  For our purposes, $B/A$ will be bounded by polynomial of $n$ and $1/\eps$, hence we may always assume that the approximation results are  monotone step functions with complexity $\tilde O(\eps^{-1})$.

For an item  set $I$ with the same profit $p_i = p$ for every item $i \in I$, the step function $f_I$ can be exactly computed in $O(n\log n)$ time by simple greedy: the function values are $0,p,2p,\dots,np$ and the $x$-breakpoints are $w_1,w_1+w_2,\,\dots,\,w_1+\dots+w_n$, after sorting all  $w_i$'s in nondecreasing order. We say that a monotone step function is \textit{$p$-uniform} if its function values are of the form $0, p, 2p, \dots , lp$ for some $l$.
We say that a $p$-uniform function is \textit{pseudo-concave} if the differences of consecutive $x$-breakpoints are nondecreasing from left to right.
In the previous case where all $p_i$'s are equal to $p$, $f_I$ is indeed $p$-uniform and pseudo-concave.

\section{Chan's techniques}

 In this section we review several useful lemmas by Chan \cite{chan2018approximation}.

\subsection{Merging profit functions}
\begin{lemma}[\text{\cite[Lemma 2(i)]{chan2018approximation}}]
\label{dc}
Let $f_1,\dots, f_m$ be monotone step functions with total complexity $O(n)$ and
ranges contained in $\{ 0\} \cup [A, B]$. Then we can compute a monotone step function that approximates $f_1 \oplus\dots \oplus f_m$ with factor $1+O(\eps)$ and complexity $\tilde O(\frac{1}{\eps}\log B/A)$ in $O(n) + \tilde O((\frac{1}{\eps})^2 m/2^{\Omega(\sqrt{\log(1/\eps)})} \log B/A)$ time.
\end{lemma}
\begin{remark}
Lemma \ref{dc} is proved using a divide-and-conquer method, which was also used previously in \cite{kellerer1999new}. The speedup of superpolylogarithmic factor $2^{\Omega(\sqrt{\log(1/\eps)})}$ is due to recent progress on $(\min,+)$-convolution \cite{bremner2014necklaces,williams2014faster,chan2016deterministic}.
\end{remark}

Lemma \ref{dc} enables us to focus on a simpler case where all $p_i \in [1,2]$. For the general case, we divide the items into $O(\log \frac{\max_j p_j}{\min_j p_j}) = O(\log \eps^{-1})$ groups, each containing items with $p_i \in [2^{j},2^{j+1}]$ for some $j$ (which can be rescaled to $[1,2]$), and finally  merge the profit functions of all groups by using Lemma \ref{dc} in $\tilde O(n+\eps^{-2})$ time.  

Assuming $\eps^{-1}$ is an integer and every $p_i\in [1,2]$, we can  round every $p_i$ down to a multiple of $\eps$, introducing only a $1+\eps$ error factor. Then there are only $m=O(\eps^{-1})$ distinct values of $p_i$. For every value of $p_i$, the corresponding profit function $f_i$ is $p_i$-uniform and pseudo-concave, and can be obtained by simple greedy (as discussed in Section \ref{secprelim}). 

\subsection{Approximating big profit values using greedy}
\label{secbiggreedy}
When all $p_i$'s are small, simple greedy gives good approximation guarantee when the answer is big enough. 
\begin{lemma}
\label{sortgreed}
 Suppose $p_i \in [1,2]$ for all $i\in I$. For $B=\Omega(\eps^{-1})$,  the function $f_I$ can be approximated with  additive error $O(\eps B)$ in $O(n\log n)$ time.
\end{lemma}
\begin{proof}
Sort the items in nonincreasing order of unit profit $p_i/w_i$. Let $\tilde f$ be the monotone step function resulting from greedy, with function values $0,p_1,p_1+p_2,\dots,p_1 + \dots + p_n$ and $x$-breakpoints $0,w_1,w_1+w_2,\dots,w_1+\dots+w_n$.  It approximates $f_I$ with an additive error of $\max_i p_i \le 2 \le O(\eps B)$ for $B=\Omega(\eps^{-1})$. 
\end{proof}
When every $p_i\in [1,2]$, we set $B=\eps^{-1}$ and let $f_H$ denote the result from greedy (Lemma \ref{sortgreed}). Then we only need to obtain a function $f_L$ which  $1+O(\eps)$ approximates $\min\{f_I,B\}$, and finally return $\max\{f_L,f_H\}$ as a  $1+O(\eps)$ approximation of $f_I$ (because when $f_I(x)$ exceeds $B$, an additive error $O(\eps B)$ implies $1+O(\eps)$ approximation factor).

\subsection{Approximation using \texorpdfstring{$\Delta$}{Delta}-multiples of small set \texorpdfstring{$\Delta$}{Delta}}

For a set $\Delta$ of numbers, we say that $p$ is a $\Delta$-multiple if it is a multiple of $\delta$ for some $\delta \in  \Delta$.

\begin{lemma}[\text{\cite[Lemma 5]{chan2018approximation}}]
\label{improvedsmawk}
Let $f_1,\dots , f_m$ be monotone step functions with ranges contained in $[0, B]$. 
Let $\Delta  \subset [\delta, 8\delta]$. If every $f_i$ is $p_i$-uniform and pseudo-concave for some
$p_i \in [1, 2]$ which is a $\Delta$-multiple, then we can compute a monotone step function that approximates 
$\min \{f_1 \oplus \dots \oplus f_m, B\}$ with additive error $O(|\Delta|\delta)$ in $\tilde O( Bm/\delta)$ time. 
\end{lemma}
\begin{remark}
An intuition of Lemma \ref{improvedsmawk} is as follows. When $p_i$'s are exact  multiples of $\delta$, standard dynamic programming algorithm maintains a result array of length $B/\delta$, and adding a new $f_i$ to the result can be done in linear time (by exploiting the pseudo-concavity of $f_i$ using the SMAWK algorithm\footnote{The SMAWK algorithm \cite{aggarwal1987geometric} finds all row-minima in an $n\times n$ matrix $A$ satisfying the \textit{Monge property} $A[i,j]+A[i+1,j+1]\le A[i,j+1]+A[i+1,j]$  using only $O(n)$ queries. }). Now if the next $p_i$ to be considered is a multiple of $\delta'\neq \delta$, we first have to round down the current results to multiples of $\delta'$, introducing an  additive error of $\delta'$. We round our results for  $|\Delta|-1$ times, so smaller $|\Delta|$ implies smaller overall additive error.
\end{remark} 

\begin{corollary}
\label{coroverynaive}
Let $f_1,\dots , f_m$ be monotone step functions with ranges contained in $ [0, B]$. 
 If every $f_i$ is $p_i$-uniform and pseudo-concave for some
$p_i \in [1, 2]$, then we can compute a monotone step function that approximates 
$\min \{f_1 \oplus \dots \oplus f_m, B\}$ with factor $1+O(\eps)$  in $\tilde O(\eps^{-1}Bm)$ time.
\end{corollary}
\begin{proof}
Assuming $\eps^{-1}$ is an integer,  adjust every $p_i$ down to the nearest multiple of $\eps$, and adjust $f_i$ accordingly. This introduces a $1+\eps$ overall error factor. Then use Lemma \ref{improvedsmawk} with $\delta=\eps, \Delta=\{\eps\}$ to compute the desired function in $\tilde O(Bm\eps^{-1})$ time.
\end{proof}

\section{Extending Chan's number-theoretic construction} 
\label{secextend}
As mentioned in Section \ref{secoutline},  the main results of this section are two approximation algorithms in $\tilde O(n+\eps^{-2}B^{1/3})$ and $\tilde O(\eps^{-3/2}n^{3/4}+\eps^{-2})$  time respectively (the latter time bound assumes $n=O(1/\eps)$). These results rely on Lemma \ref{improvedsmawk}. 
\subsection{Number-theoretic construction}
\label{secconstruct}

To avoid checking boundary conditions, from now on we assume $\eps>0$ is sufficiently small.

Our algorithm extends Chan's technique by using a multi-level structure defined as follows. 
\begin{definition}
For fixed parameters $\delta_1,\delta_2,\dots,\delta_d$ satisfying condition
\begin{equation}
\label{eqnparacondition}
\eps \le \delta_1,\,\, \delta_i\le \delta_{i+1}/2,\,\, \delta_d\le 1/8
\end{equation}
and a finite real number set $\Delta_1\subset [\delta_1,8\delta_1]$, a \emph{set tower} $(\Delta_1,\Delta_2,\dots,\Delta_d)$ generated by $\Delta_1$ is defined by recurrence\footnote{For a number $k$ and a set $A$ of numbers, $kA:= \{ka: a \in A\}$.}
\begin{equation}
\Delta_{i+1}:=[\delta_{i+1},8\delta_{i+1}] \cap \bigcup_{k\in \Z} k\Delta_i,\,\,\,\, i=1,2,\dots,d-1.
\end{equation}
We refer to $\Delta_1$ as the \emph{base set} and $\Delta_d$ as the \emph{top set} of this set tower. We also say that the base set $\Delta_1$ \textit{generates} the top set $\Delta_d$.

If $\Delta_d^\ast$ is the top set generated by a singleton base set $\Delta_1^\ast = \{x\}$, then for every $y\in \Delta_d^\ast$ we say $x$ \emph{generates} $y$. 
\label{defsettower}
\end{definition}
We have the following simple facts about set towers.
\begin{prop}
\label{propmulti}
 If $x$ generates $y$ then $x\in \Delta_1$ implies $y \in \Delta_d$. Conversely, for every $y\in \Delta_d$, there exists $x\in \Delta_1$ which generates $y$, and for every  $1\le i\le d$ there exists $z\in \Delta_i$ such that both $y/z$ and $z/x$ are integers.
\end{prop}

\begin{prop}
\label{propbound}
For any $1\le i\le d$, $|\Delta_{i}| \le 8^{i-1}(\delta_i/\delta_1)  |\Delta_1|$, and we can compute $\Delta_i$ in $\tilde O( 8^{i-1}(\delta_i/\delta_1)  |\Delta_1|)$ time given $\Delta_1$ as input.
\end{prop}
\begin{proof}
For $2\le i\le d$, we have 
\begin{align*}
|\Delta_i|= \Big\lvert [\delta_i,8\delta_{i}] \cap \bigcup_{k\in \Z}k\Delta_{i-1}\Big \rvert 
\le \sum_{x\in \Delta_{i-1}} 8\delta_i/x \le |\Delta_{i-1}| 8\delta_i/\delta_{i-1}.
\end{align*}
The proof of size upper bounds follows by induction. Elements of $\Delta_i$ can be generated straightforwardly within the time bound.
\end{proof}

\begin{lemma}
\label{lemmanumbertheory}
Let $T_1,T_2,\dots,T_d$ be positive real numbers satisfying $T_1\ge 2$ and $T_{i+1}\ge 2T_{i}$. There exist at least $T_d\big /(\log T_d)^{O(d)}$ integers $t$ satisfying the following condition: $t$ can be written as a product of integers $t=n_1n_2\cdots n_d$, such that $n_1n_2\cdots n_i\in (T_i/2,T_i]$ for every $1\le i \le d$. 
\end{lemma}
The proof of Lemma \ref{lemmanumbertheory} is deferred to Appendix \ref{apdnumbertheory}. Lemma \ref{lemmanumbertheory} helps us prove the following fact, which is a multi-level extension of \cite[Lemma 6]{chan2018approximation}.

\begin{lemma}
\label{lemmadense}
For any parameters $\delta_1,\dots,\delta_d$ satisfying condition (\ref{eqnparacondition}),  
there exists a base set $\Delta_1$ of size $\frac{\delta_1}{\eps}\cdot (\log \eps^{-1})^{O(d)}$, such that every $p\in [1,2]$ can be approximated by a $\Delta_d$-multiple with additive error $O(\eps)$, where $\Delta_d$ is the top set generated by $\Delta_1$.

This base set $\Delta_1$ can be constructed in $\tilde O(\eps^{-1}\delta_1^{-1})$ time deterministically.
\end{lemma}

\begin{proof}
Let $P=\{1,1+\eps,1+2\eps,\dots, 1+\lfloor \frac{1}{\eps}\rfloor \eps\}$. It suffices to approximate every value $p\in P$ with additive error $\eps$ using $\Delta_d$-multiples. 
For any $p\in P$ and $y\in \Delta_d \subset [\delta_d,8\delta_d]$,  $p$ is approximated with additive error $\eps$ by a multiple of $y$ if and only if  $y \in \bigcup_{j\in \Z} \left [\frac{p-\eps}{j},\frac{p}{j}\right ]$. 

  Our constructed base set $\Delta_1$ will satisfy $\Delta_1\subset [\delta_1,4\delta_1]$.
 Suppose integers $k_{1},k_{2},\dots,k_{d-1}$ satisfy \begin{equation}
 \label{eqnconda}
 k_{1}k_{2}\cdots k_{i-1} \in [\delta_i/\delta_{1}, 2\delta_i/\delta_1],\,\,  \text{ for every $2\le i \le d$}.                                                     \end{equation}
Then by Definition \ref{defsettower}, for any $x\in \Delta_1 \subset [\delta_1,4\delta_1]$, we have  $xk_{1}k_2\cdots k_{i-1} \in \Delta_i$ for every $2\le i \le d$.
  
  For any integer $j$ satisfying \begin{equation}
 \label{eqncondb}
k_1k_2\cdots k_{d-1}j \in [p/(4\delta_1),p/(2\delta_1)],
\end{equation}
the interval $[\frac{p-\eps}{k_{1}k_2\cdots k_{d-1}j }, \frac{p}{k_{1}k_2\cdots k_{d-1}j }]$ is contained in $[\delta_1,4\delta_1]$. 

We say an integer $K$ is \emph{good} for $p$, if $K$  can be expressed as a product of integers $k_1 k_2\cdots k_{d-1} j$ satisfying conditions (\ref{eqnconda}) and (\ref{eqncondb}). For such $K$, any $x\in [\frac{p-\eps}{K},\frac{p}{K}] \cap \Delta_1$  generates an element $y=xk_1k_2\cdots k_{d-1} \in \Delta_d \cap [\frac{p-\eps}{j},\frac{p}{j}]$ such that $p$ can be approximated by a multiple of $y$ with additive error $\eps$.

By Lemma \ref{lemmanumbertheory}, the number of good integers $K$ for $p$ is at least
\begin{equation*}
\frac{p/(4\delta_1)}{\big (\log (p/(4\delta_1))\big  )^{O(d)}} = \Omega \Big ( \frac{\delta_1^{-1}}{(\log \eps^{-1})^{O(d)}}\Big ),
\end{equation*}
and at most $p/(2\delta_1) = O(\delta_1^{-1})$, by (\ref{eqncondb}).
Using conditions (\ref{eqnconda}) and (\ref{eqncondb}) we can compute all these $K$'s by simple dynamic programming.
We denote the union of their  associated intervals by \begin{equation}
I_p \,\,:= \bigcup_{\text{$K$ good for $p$}}\left [\frac{p-\eps}{K},\frac{p}{K}\right ]  \,\,\,\,\subset \,\big [\delta_1,4\delta_1\big ].   
\end{equation} 
Note that these intervals are disjoint since $p/(K+1) \le (p-\eps)/K$, so the total length of $I_p$ can be lower-bounded as   
\begin{equation}
\lambda (I_p) \ge \frac{\delta_1^{-1}}{(\log \eps^{-1})^{O(d)}} \cdot \frac{p-(p-\eps)}{\max K}  \ge \frac{\eps}{(\log \eps^{-1})^{O(d)}}.
\end{equation}

We have seen that $p$ is approximated by a $\Delta_d$-multiple with additive error $\eps$ as long as $\Delta_1 \cap I_p \neq \emptyset$. 
We compute $I_p$ for every $p\in P$, and use the standard greedy algorithm (for Hitting Set problem) to construct a base set $\Delta_1\subset [\delta_1,4\delta_1]$ which intersects with every $I_p$: in each round we find 
a point $x\in [\delta_1,4\delta_1]$ that hits the most $I_p$'s, include $x$ into $\Delta_1$, and remove the $I_p$'s that are hit by $x$. In each round the number of remaining $I_p$'s decreases by 
\begin{equation*}
s:=\frac{\min_p \lambda(I_p)}{4\delta_1-\delta_1} \ge \frac{\eps/\delta_1}{(\log \eps^{-1})^{O(d)}},
\end{equation*}
so the solution size $|\Delta_1|$ is upper-bounded by 
\begin{equation*}
 1+ \log_{1/(1-s)}|P| = O\left (\frac{\log |P| }{s}\right ) =\frac{\delta_1}{\eps} (\log\eps^{-1})^{O(d)}.
\end{equation*}
To implement this greedy algorithm, we use standard data structures (for example, segment trees) that support finding $x$ which hits the most intervals, reporting an interval hit by $x$, removing an interval, all in logarithmic time per operation. The number of operations is bounded by the total number of small intervals, so the running time is at most $\tilde O( |P| \cdot \frac{p}{2\delta_1}) = \tilde O(\delta_1^{-1}\eps^{-1})$. 
\end{proof}

The following lemma evenly partitions the base set $\Delta_1$ into $r$ subsets $\Delta_1^{(1)},\dots,\Delta_1^{(r)}$, and partitions the profit values $P=\{p_1,\dots,p_m\}$ into $P^{(1)}\cup \dots \cup P^{(r)}$, so that $P^{(j)}$ can be approximated by $\Delta_d^{(j)}$-multiples. An additional requirement is that $P^{(1)},\dots,P^{(r)}$ should have size $O(|P|/r)$ each.
\begin{lemma}
\label{lemmagrouping}
Let $\delta_1,\dots,\delta_d$ be parameters satisfying condition (\ref{eqnparacondition}). 
Let $P = \{p_1,\dots,p_m\}  \subset [1,2]$ with $m= O(\eps^{-1})$. 
 Given a positive integer parameter $r =O(\min \{\frac{\delta_1}{\eps},m \})$, there exist $r$ base sets $\Delta_{1}^{(1)},\Delta_{1}^{(2)},\dots,\Delta_{1}^{(r)}$ each of size $\frac{\delta_1}{\eps r} \cdot (\log \eps^{-1})^{O(d)}$, and a partition of $P = P^{(1)} \cup  P^{(2)} \cup\cdots \cup P^{(r)}$ each of size $O(m/r)$, such that for every $1\le j \le r$, every $p \in P^{(j)}$ can be approximated by a $\Delta_{d}^{(j)}$-multiple with additive error $O(\eps)$, where $\Delta_{d}^{(j)}$ is the top set  generated by $\Delta_1^{(j)}$. 

These $r$ base sets and the partition of $P$ can be computed by a  deterministic algorithm in $\tilde O(\eps^{-2}/r)$ time .
\end{lemma}
\begin{proof}
First construct the base set $\Delta_1$ of size $\frac{\delta_1}{\eps} (\log \eps^{-1})^{O(d)}$ from Lemma \ref{lemmadense} in $\tilde O(\delta_1^{-1}\eps^{-1}) =\tilde O(\eps^{-2}/r)$ time, and compute the top set $\Delta_d$ that it generates. By Proposition \ref{propbound}, $|\Delta_d| \le 8^{d-1} \frac{\delta_d}{\delta_1}|\Delta_1|\le \frac{\delta_d}{\eps} (\log \eps^{-1})^{O(d)}$.
Generate and sort all $\Delta_d$-multiples in interval $[1,2]$. For every $p\in P$, use binary search  to find the $\Delta_d$-multiple $ky \le p\, (y\in \Delta_d)$ closest to $p$, and then add $p$ to the set $Q_x$, where $x\in \Delta_1$ is an element that generates $y$. ($Q_x$ is initialized as empty for every $x\in \Delta_1$.)  Then remove every $x$ with $Q_x=\emptyset$ from $\Delta_1$, so that $|\Delta_1|\le m$, while every $p\in P$ can still be approximated with $O(\eps)$ additive error by a $\Delta_d$-multiple.

Let $D:=\max\{r,|\Delta_1|\}$ and  let
 $s: = \lceil m/D\rceil$. For every $x\in \Delta_1$ we divide $Q_x$ evenly into small subsets each having size at most $s$. The total number of  these small subsets is 
\begin{equation*}
 \sum_{x\in \Delta_1}\lceil |Q_x|/s \rceil \le |\Delta_1| + \sum_{x \in \Delta_1} |Q_x|/s = |\Delta_1| + m/s \le 2 D.
\end{equation*}
We merge these small subsets into $r$ groups, each having at most $ \lceil 2D/r \rceil$ small subsets.  Then, define set $P^{(j)}$ as the union of small subsets from the $j$-th group, and let base set $\Delta_{1}^{(j)}$ contain $x\in \Delta_1$ if  any of these small subsets comes from $Q_x$. So $|\Delta_{1}^{(j)}| \le \lceil 2D/r \rceil =\frac{\delta_1}{\eps r} ( \log \eps^{-1})^{O(d)}$, and $|P^{(j)}| \le s\cdot  \lceil 2D/r\rceil = O(m/D) \cdot O(D/r) = O(m/r)$. 
\end{proof}

\subsection{Approximation using set towers }
\label{secapprox}
We first prove a slightly improved version of Corollary \ref{coroverynaive}. The only purpose of this lemma is to get rid of the $(\log \eps^{-1})^{O(\log \log \eps^{-1})}$ factor in the final running time.
\begin{lemma}
\label{naive} 
Let $f_1,\dots , f_m$ be monotone step functions with ranges contained in $ [0, B]$ for some $1\le B \le O(\eps^{-1})$. 
 If every $f_i$ is $p_i$-uniform and pseudo-concave for some
$p_i \in [1, 2]$, then we can compute a monotone step function that approximates 
$\min \{f_1 \oplus \dots \oplus f_m, B\}$ with factor $1+O(\eps)$  in $\tilde O(\eps^{-1}(Bm+\eps^{-1})/B^{0.01})$ time.
\end{lemma}
\begin{proof}

Using Lemma \ref{lemmadense} with parameters $d=1,\delta_1 = \eps B^{0.01}$, we get $\Delta \subset [\delta_1,8\delta_1]$ with size $|\Delta| \le \tilde O(\delta_1/\eps) = \tilde O(B^{0.01})$, in $\tilde O(\eps^{-2}/B^{0.01})$ time. Adjust every $p_i$ down to the nearest $\Delta$-multiple, and adjust $f_i$ accordingly. This introduces at most $1+O(\eps)$ error factor. Then use Lemma \ref{improvedsmawk}  to compute a monotone step function $f_H$ that approximates $\min \{f_1\oplus \cdots \oplus f_m, B\}$ with additive error $e=O(|\Delta|\delta_1) = \tilde O(\eps B^{0.02})$, in $\tilde O(B^{0.99}m\eps^{-1})$ time. 

Let $B_L:=e/\eps$, and use Corollary \ref{coroverynaive} to compute a monotone step function $f_L$ that approximates $\min \{f_1\oplus \cdots \oplus f_m, B_L\}$ with factor $1+O(\eps)$ in only $\tilde O(B_Lm\eps^{-1}) = \tilde O(B^{0.02} m\eps^{-1})$ time.

Since $f_H$ approximates $\min \{f_1\oplus \cdots \oplus f_m,B\}$ with additive error $\eps B_L$, $\max\{f_L,f_H\}$ is a $1+ O(\eps)$ approximation of $\min \{f_1\oplus \cdots \oplus f_m, B\}$. 
\end{proof}

Now we can approximate the profit function $\min \{B, \bigoplus_{p_k \in P^{(j)}} f_k\}$ for each group $P^{(j)}$, using the multi-level approach described in Section~\ref{secoutline}.
\begin{lemma}
\label{lemmaeachgroup}
  Let $f_1,\dots, f_m$ be given monotone step functions with ranges contained in $[0, B]$, and every $f_k$ is $p_k$-uniform and pseudo-concave for some $p_k \in [1, 2]$.  Assume  $m=  O(\eps^{-1})$, $\Omega(\eps^{-0.01}) \le B \le O(\eps^{-1})$.  Let $r$ be a given positive integer parameter with $r= O(m), r=o(B)$.
  
  We can set $d=O(\log \log \eps^{-1})$ and choose $d$ parameters $\delta_1,\dots,\delta_d$ satisfying condition (\ref{eqnparacondition}), such that the following statement holds:
  
  Let $P^{(1)}\cup \dots \cup P^{(r)}$ be the partition of set $P=\{p_1,\dots,p_m\}$ returned by the algorithm in Lemma \ref{lemmagrouping} with the above parameters.  Then for any $1\le j\le r$, using the base set $\Delta_{1}^{(j)}$ from Lemma \ref{lemmagrouping}, we can compute a monotone step function that approximates $\min \{B, \bigoplus_{p_k \in P^{(j)}} f_k\}$ with factor $1+O(\eps)$, in $(\eps^{-2}/r^{0.01}+m\eps^{-1} B^{1/2}/r^{3/2})(\log \eps^{-1})^{O(d)}$ time. 
\end{lemma}
\begin{proof}
We can assume $B\ge 4r$, and define $d$ to be the unique positive integer such that 
\begin{equation*}
 2^{2^{d-1}} \le \frac{\sqrt{B}}{\sqrt{r}} <2^{2^{d}} = 4^{2^{d-1}}.
\end{equation*}
Then $d= O(\log \log \frac{\sqrt{B}}{\sqrt{r}}) = O(\log\log \eps^{-1})$. Pick $\alpha\in [2,4)$ such that
\begin{equation}
\alpha^{2^{d-1}} = \frac{\sqrt{B}}{\sqrt{r}}.
\end{equation}
Define 
\begin{equation}
 \delta_{i} := \eps \sqrt{Br}\Big \slash {\alpha^{2^{d-i}}} ,\,\, 0\le i \le d.
\end{equation}
Then 
\begin{equation}
 \delta_d = \frac{\eps\sqrt{Br}}{\alpha},\, \delta_1 = \eps r
\end{equation}
Note that $\delta_d = \eps \sqrt{B}\cdot O(\sqrt{r}) = \eps \sqrt{B} \cdot o(\sqrt{B})=\eps \cdot o(B)= o(1)$. Hence the parameters $\delta_1,\dots,\delta_d$ satisfy condition (\ref{eqnparacondition}) for sufficiently small $\eps$.  

The base set $\Delta_{1}^{(j)}$ from Lemma \ref{lemmagrouping} has size $\frac{\delta_1}{\eps r} (\log \eps^{-1})^{O(d)}$. We compute the generated set tower $\Delta_{1}^{(j)}, \Delta_{2}^{(j)},\dots, \Delta_d^{(j)}$.
By Proposition \ref{propbound}, $|\Delta_{i}^{(j)}| \le  \frac{\delta_i}{\eps r} (\log \eps^{-1})^{O(d)}$. 
   Let\begin{equation}
 t:= \max\Big \{\alpha,\,\,\max_{j} |\Delta_{i}^{(j)}|\Big /\frac{\delta_i}{\eps r} \Big \}= (\log \eps^{-1})^{O(d)}           
	  \end{equation}
and define 
\begin{equation}
 B_i: = Bt\Big / \alpha^{2^{d-i}}, \,\,0\le i \le d.
\end{equation} 
Then $B\le B_d\le  B \cdot (\log \eps^{-1})^{O(d)}$, and it's easy to verify that 
 \begin{equation}
 |\Delta_{i}^{(j)}| \cdot \delta_i \le B_{i-1}\eps, \,\,\,(1\le i \le d). 
 \end{equation}

For every $1\le i\le d$,  adjust every $p_k \in P^{(j)}$ down to the nearest $\Delta_{i}^{(j)}$-multiple and adjust $f_k$ accordingly, which introduces a $1+O(\eps)$  error factor.
Then use Lemma \ref{improvedsmawk} to obtain a monotone step function $g_i$ which approximates $\min \{\bigoplus_{p_k\in P^{(j)}} f_k, B_i\}$  with additive error $O(|\Delta_{i}^{(j)}|\delta_i) = O( \eps B_{i-1})$ in $\tilde O(|P^{(j)}|B_i/\delta_{i})$ time. 

Then we use Lemma \ref{naive} to obtain a monotone step function $g_0$ which  approximates $\min \{\bigoplus_{p_k\in P^{(j)}} f_k, B_{0}\}$ with $1+O(\eps)$ factor, in $\tilde O(\eps^{-1}(|P^{(j)}|B_{0}+\eps^{-1})B_0^{-0.01} )$ time. Notice that $B_0=rt$.

Finally, $\max\{g_0,g_1,g_2,\dots,g_d\}$ is a  $1+O(\eps)$ approximation of $\min\{\bigoplus_{p_k\in P^{(j)}}f_k,B_d\}$, where $B_d\ge B$. Overall running time is 
\begin{align*}
& \tilde O(\eps^{-1}(|P^{(j)}|B_{0}+\eps^{-1})B_0^{-0.01} )+\sum_{1\le j\le d}\tilde O(|P^{(j)}|B_j/\delta_j) \\
= \ & \tilde O\big (\eps^{-1}(\frac{m}{r} \cdot (rt) + \eps^{-1})(rt)^{-0.01} \big  ) + d\cdot \tilde O\big ( \frac{m}{r} B_d/\delta_d\big ) \\
= \ & (\eps^{-2}/r^{0.01}+m\eps^{-1} B^{1/2}/r^{3/2})(\log \eps^{-1})^{O(d)}.
\end{align*}
\end{proof}

Now we merge the results from all $r$ groups, and obtain an approximation of the final result $\min\{f_1\oplus\dots \oplus f_m,B\}$.
\begin{lemma}
\label{lemmadependsonB}
  Let $f_1,\dots, f_m$ be given monotone step functions with ranges contained in $[0, B]$, and every $f_k$ is $p_k$-uniform and pseudo-concave for some $p_k \in [1, 2]$.  Assume  $m=  O(1/\eps), \Omega(\eps^{-0.01})\le B \le O(\eps^{-1})$.   We can approximate $\min \{f_1\oplus \dots \oplus f_m,B\}$ with factor $1+O(\eps)$ in $O( \eps^{-2} B^{1/3}/2^{\Omega(\sqrt{\log(1/\eps)})} )$ time.
\end{lemma}
\begin{proof}
Assume $m\ge \eps^{-1}$, by adding  zero functions which do not change the answer.

Let $r = o(B)$ be a positive integer parameter to be determined later.

Using Lemma \ref{lemmagrouping} and Lemma \ref{lemmaeachgroup}, we can get a partition of $\{p_1,\dots,p_m\} = P^{(1)}\cup\dots \cup P^{(r)}$ and then get an $1+O(\eps)$ approximation of $\min\{\bigoplus_{p_k\in P^{(j)}}f_k,B\}$ for every $1\le j \le r$, in $r\cdot (\eps^{-2}/r^{0.01}+ m\eps^{-1}B^{1/2}/r^{3/2}) (\log \eps^{-1})^{O(d)} =(r^{0.99}+\sqrt{B/r}) \eps^{-2} (\log \eps^{-1})^{O(\log \log \eps^{-1})}$ overall time.

Then we use Lemma \ref{dc} to  merge all these $r$ functions in  $\tilde O((\frac{1}{\eps})^2 r/2^{\Omega(\sqrt{\log(1/\eps)})})$ time.

Setting $r = B^{1/3} 2^{c\sqrt{\log(1/\eps)}}$, where $c>0$ is some small enough constant, the total complexity is 
$$  O( \eps^{-2} B^{1/3}/2^{\Omega(\sqrt{\log(1/\eps)})} ). $$
\end{proof}

\begin{lemma}
\label{lemmadependsonm}
Let $I$ be a set of $m$ items with $p_i \in [1,2]$ for every $i\in I$, where $\Omega(\eps^{-2/3}) \le m \le  O(\eps^{-1})$.   One can approximate $f_I$ with factor $1+O(\eps)$ in 
$  O( \eps^{-{3/2}} m^{3/4}/2^{\Omega(\sqrt{\log(1/\eps)})} ) $
 time.
\end{lemma}
\begin{proof}
Let $f_1,\dots, f_m$ denote the profit functions of the $m$ items.

Let $r = o(m^{1/2})$ be a positive integer parameter to be determined later.  
Obtain a partition of $\{p_1,\dots,p_m\} = P^{(1)}\cup\dots \cup P^{(r)}$ using Lemma \ref{lemmagrouping}. Let $B:= \max_i \sum_{p\in P^{(i)}}p \le 2\max_i |P^{(i)}| = \Theta(m/r)$.
Then $r  = o(B)$. Use Lemma \ref{lemmaeachgroup} to get an $1+O(\eps)$ approximation of $\bigoplus_{p_k\in P^{(j)}}f_k = \min\{\bigoplus_{p_k\in P^{(j)}}f_k,B\}$ for every $1\le j \le r$, in $r\cdot (\eps^{-2}/r^{0.01}+m\eps^{-1}B^{1/2}/r^{3/2} ) (\log \eps^{-1})^{O(d)} =(\eps^{-2}r^{0.99}+m^{3/2} \eps^{-1}/r )(\log \eps^{-1})^{O(\log \log \eps^{-1})}$ overall time.

Then we use Lemma \ref{dc} to  merge all these $r$ functions in  $\tilde O((\frac{1}{\eps})^2 r/2^{\Omega(\sqrt{\log(1/\eps)})})$ time.
 
Setting $r = m^{3/4}\eps^{1/2} 2^{c\sqrt{\log(1/\eps)}}$, where $c>0$ is some small enough constant, the total complexity is 
$$  O( \eps^{-{3/2}} m^{3/4}/2^{\Omega(\sqrt{\log(1/\eps)})} ). $$
\end{proof}

\begin{corollary}[restated Theorem \ref{mainsmall}]
\label{small}
  For $n = O(\frac{1}{\eps})$, there is a deterministic $(1+\eps)$-approximation algorithm for 0-1 knapsack in $O\Big ( \big (n^{3/4}(\frac{1}{\eps})^{3/2} + (\frac{1}{\eps})^2\big ) /2^{\Omega(\sqrt{\log(1/\eps)})}\Big)$ time.
\end{corollary}

\begin{proof}
 Divide the items into $O(\log \frac{n}{\eps})$ groups, each containing items with $p_i \in [2^{j},2^{j+1}]$ for some $j$. Use Lemma \ref{lemmadependsonm} to solve each group, and merge them using Lemma \ref{dc}. 
\end{proof}

\section{Main algorithm}
\label{secmain}
\subsection{A greedy lemma}
Our improved algorithm uses the following lemma, which gives an upper bound on the total profit of cheap items (with low $p_i/w_i$) in an optimal knapsack solution.

\begin{lemma}
\label{greedy}
Let $H,L$ be two subsets of items with $p_i\in [1,2]$.  Let $W = \sum_{h \in H} w_h$ and $q = \min_{h\in H} \frac{p_h}{w_h}$. Suppose $\max_{l \in L} \frac{p_l}{w_l} \le q(1-\alpha)$ for some $0<\alpha < 1$.  Let $f = f_H \oplus f_L, \tilde f = f_H \oplus \min\{\frac{2}{\alpha}, f_L\}$. Then for every $ x \le W$, $f(x)= \tilde f(x)$.
\end{lemma}
\begin{proof}
By greedy, $f(W)=\sum_{h\in H}p_h = \tilde f(W)$ clearly holds. Now consider $0\le x<W$. Suppose $f_L(x')+f_H(x-x')$ achieves its maximum value at $x'= w_L$, i.e.,  $f(x) =f_L(w_L) + f_H(x-w_L)$. It suffices to prove $f_L(w_L) \le \frac{2}{\alpha}$.

Let $J \subseteq H$ be a subset of items with total weight $w_J\le x-w_L$ and total profit achieving optimal value $f_H(x-w_L)$. Let $K\subseteq H \backslash J$ be a subset of items with total weight $w_K$, such that  $w_K \le w_L$, and $w_K+w_i>w_L$ for every remaining item $i \in H\backslash (J\cup K)$. Such $K$ can be constructed by a simple greedy algorithm. 
  
  Since $w_J+w_K \le (x-w_L)+w_L < W = \sum_{h\in H}w_h$, the remaining set $H\backslash (J\cup K)$ contains at least one item $h_0$. Hence, $w_L-w_K< w_{h_0} = p_{h_0} /\frac{p_{h_0}}{w_{h_0}} \le 2/q$, and equivalently $qw_K>qw_L-2$.

 Since $J\cup K$ is a subset of $H$ with total weight bounded by $x$, we have $f_H(x) \ge \sum_{k \in K}p_k +\sum_{j \in J}p_j$, and thus $f_H(x)-f_H(x-w_L) = f_H(x) - \sum_{j\in J}p_j \ge \sum_{k\in K} p_k \ge qw_K >qw_L-2$.

Hence $qw_L-2 < f_H(x) - f_H(x-w_L) \le f(x) - f_H(x-w_L) = f_L(w_L)  \le  q(1-\alpha)w_L$, which shows that $q\alpha w_L \le 2$. So $f_L(w_L) \le q(1-\alpha)w_L \le qw_L \le 2/\alpha$, which concludes the proof.
\end{proof}

\subsection{FPTAS for Subset Sum}
We will use the efficient FPTAS for the subset sum problem by Kellerer et al. \cite{kellerer2003efficient} as a subroutine in our algorithm.

\begin{lemma}[\cite{kellerer2003efficient}, implicit]
\label{lemmakellerer}
Let $I$ be a set of $n$ items and $W$ be a number. We can obtain a list $S$ of $O(\frac{1}{\eps})$ numbers in $O(n + (\frac{1}{\eps})^{2}\log\frac{1}{\eps})$ time, such that for every $s\le W$ that is the subset sum $s=\sum_{j\in J}w_j$ of some subset $J\subseteq I$, there exists $s'\in S$ with $s-\eps W \le s'\le s$.
\end{lemma}
\begin{remark}
This result wasn't explicitly stated in \cite{kellerer2003efficient}, but can be easily seen from their analysis of the correctness of the FPTAS.
\end{remark}
%TODO verify this

\begin{corollary}
\label{subsetsum}
Let $I$ be a set of $n$ items with $p_i \in [1,2]$ and $p_i=w_i$ for every item $i \in I$.  We can approximate $f_I$ with factor $1+O(\eps)$ in $O(n\log n+\eps^{-2} \log \frac{1}{\eps}\log n)$ time.
\end{corollary}
\begin{proof}
Notice that approximating $s$ with additive error $\eps W$ implies approximation factor $1+O(\eps)$ for $W/2\le s \le W$. So we simply apply Lemma \ref{lemmakellerer} with $W = 2^j$ for $0\le j \le 1+\log n$, and merge all obtained lists. 
\end{proof}

\subsection{Improved algorithm}

\begin{lemma}
\label{lemmaimproved}
Let $I$ be a set of $n$ items with $p_i \in [1,2]$ for every $i \in I$. We can approximate $f_I$ with factor $1+O(\eps)$ in $O(n \log \frac{1}{\eps} + (\frac{1}{\eps})^{9/4}/2^{\Omega(\sqrt{\log(1/\eps)})})$ time.
\end{lemma}
\begin{proof}

Let $B= \lceil \eps^{-1} \rceil$ and  assume $n \ge B$ (if $n<B$, we can simply apply Lemma \ref{lemmadependsonm}). By Lemma \ref{sortgreed}, we can approximate $f_I$ with additive error $O(\eps B)$ in $O(n \log \frac{1}{\eps})$ time, so we only need to approximate $\min\{f_I,B\}$ with factor $1+O(\eps)$.

We sort the items by their unit profits $p_i/w_i$. Let set $H$ contain the top $B$ items with the highest unit profits. Define $q = \min_{h \in H}\frac{p_h}{w_h}$, and let $M$ be the set of remaining items $i$ with $q(1-\alpha) \le \frac{p_i}{w_i} \le q$, where parameter $0<\alpha<1$ is to be determined later. Let set $L$ contain the remaining items not included in $H$ or $M$.

Using Lemma \ref{lemmadependsonm}, we can compute $\tilde f_H$ which approximates $f_H$ with factor $1+O(\eps)$ in time $O(B^{3/4}\eps^{-3/2}/2^{\Omega(\sqrt{\log(1/\eps)})}) = O(\eps^{-9/4}/2^{\Omega(\sqrt{\log(1/\eps)})})$.

Since $\max_{l \in L}\frac{p_l}{w_l}<q(1-\alpha)$, Lemma \ref{greedy} states that $f_H \oplus f_L$ and $f_H \oplus \min \{2/\alpha, f_L\}$ agree when $x\le W_H =   \sum_{h \in H}w_h$. Since $(f_H\oplus f_L)(W_H)  = \sum_{h\in H} p_h \ge B$, this implies $\min\{B,f_H \oplus f_L\} = \min\{B, f_H \oplus \min\{2/\alpha, f_L\}\}$. For every item $l \in L$, we round down $p_l$ to a power of $1+\eps$, so that there are only $\log_{1+\eps}2 = O(\eps^{-1})$ distinct values. This only multiplies the approximation factor by $1+\eps$.  Then we use Lemma \ref{lemmadependsonB} to compute an approximation of $\min \{2/\alpha,f_L\}$ with factor $1+O(\eps)$ in $\tilde O(\eps^{-2} (2/\alpha)^{1/3}/2^{\Omega(\sqrt{\log(1/\eps)})})$ time. We merge it with $\tilde f_H$ and obtain an approximation of $\min \{f_H\oplus f_L, B\}$ with factor $1+O(\eps)$.

For every $m\in M$, we round down $p_m$ so that the unit profit $p_m/w_m$ becomes a  power of $1+\eps$. After rounding, the approximation factor is only multiplied by $1+\eps$, and there are at most $\log_{1+\eps}\frac{q}{q(1-\alpha)} = O(\alpha/\eps)$ distinct unit profits in $M$. Let $M_q$ denote the set of items in $M$ with unit profit $q$. For each $q$, we use Lemma \ref{subsetsum} to obtain a $1+\eps$ approximation of the function $f_{M_q}$ in $O(|M_q|+\eps^{-2})$ time. Then we use Lemma \ref{dc} to merge these functions and obtain a $1+\eps$ approximation of $f_M$. 
The total time is $O(|M| \log n) + \tilde O(\alpha\eps^{-3})$.

Finally we merge the functions and get an approximation of $\min \{ B, f_L\oplus f_H \oplus f_M\}$ with factor $1+O(\eps)$.
The total time is $O(n\log\frac{1}{\eps}) + \tilde O(\alpha \eps^{-3} + \eps^{-2} (2/\alpha)^{1/3}/2^{\Omega(\sqrt{\log(1/\eps)})})$, which is $ O(n\log\frac{1}{\eps} +\eps^{-9/4}/2^{\Omega(\sqrt{\log(1/\eps)})})$ if we choose $\alpha = \eps^{3/4}/2^{c\sqrt{\log(1/\eps)}}$ for a sufficiently small constant $c$. 
\end{proof}
\begin{corollary}[restated Theorem \ref{mainthm}]
  There is a deterministic $(1+\eps)$-approximation algorithm for 0-1 knapsack with running time $O(n \log \frac{1}{\eps} + (\frac{1}{\eps})^{9/4} /2^{\Omega(\sqrt{\log(1/\eps)})})$.
\end{corollary}
\begin{proof}
 Divide the items into $O(\log \frac{n}{\eps})$ groups, each containing items with $p_i \in [2^{j},2^{j+1}]$ for some $j$.
Use Lemma \ref{lemmaimproved} to solve each group, and merge them using Lemma \ref{dc}. 
\end{proof}

%%
%% Bibliography
%%

%% Please use bibtex, 

\bibliography{lipics-v2019-sample-article}

\appendix
\section{Proof of Lemma \ref{lemmanumbertheory}}
\label{apdnumbertheory}

\begin{theorem}[Reminder of Lemma \ref{lemmanumbertheory}]
Let $T_1,T_2,\dots,T_d$ be positive real numbers satisfying $T_1\ge 2$ and $T_{i+1}\ge 2T_{i}$. There exist at least $T_d\big /(\log T_d)^{O(d)}$ integers $t$ satisfying the following condition: $t$ can be written as a product of integers $t=n_1n_2\cdots n_d$, such that $n_1n_2\cdots n_i\in (T_i/2,T_i]$ for every $1\le i \le d$. 
\end{theorem}
\begin{proof}

For every $1\le k \le d$, we say an ordered $k$-tuple $(p_1,p_2,\dots,p_k)$ is \emph{valid} if every $p_i$ is prime, and $p_1p_2\cdots p_i \in (T_i/2,T_i]$ for every $1\le i \le k$. Then the product $t=p_1p_2\cdots p_d$ of any valid $d$-tuple $(p_1,\dots,p_d)$  satisfies our condition. For any integer $t$, there are at most $d!$ different valid $d$-tuples with product $t$ (which could be obtained by permuting $t$'s prime factors).  Let $N_k$ denote the number of valid $k$-tuples.
Then it suffices to show $N_d/(d!) \ge T_d/(\log T_d)^{O(d)}$.

By the prime number theorem and  Bertrand-Chebyshev theorem, there exists a positive constant $C$ such that $$\pi(x)-\pi (x/2 ) \ge x/(C\log x),\,\,\text{for all $x\ge 2$,}$$ where $\pi(x)$ denotes the number of primes less than or equal to $x$. 
 We will prove  $N_k \ge T_k/(C\log T_k)^k$ for all $1\le k\le d$ by induction. 
 
First note that this statement is trivial for $k=1$. 
For $k\ge 2$, a valid $k$-tuple $(p_1,\dots,p_k)$ can be obtained by appending any prime $p_k\in \big (T_k/(2P), T_k/P\big ]$ to any valid $(k-1)$-tuple $(p_1,\dots,p_{k-1})$ with product $P=p_1\cdots p_{k-1} \le T_{k-1}$. The number of such primes $p_k$ is $$\pi(T_k/P)-\pi\big (T_k/(2P)\big )\ge \frac{T_k/P}{C\log (T_k/P)} \ge \frac{T_k/T_{k-1}}{C\log T_k}.$$ 
Summing over all valid $(k-1)$-tuples, we have
$$N_k \ge N_{k-1} \cdot \frac{T_k/T_{k-1}}{C\log T_k} \ge \frac{T_{k-1}}{(C\log T_{k-1})^{k-1}} \cdot \frac{T_k/T_{k-1}}{C\log T_k}\ge \frac{T_k}{(C\log T_k)^{k}}.$$

Hence, $N_d\ge T_d/(C\log T_d)^d$ by induction. Observe that $T_d\ge 2^{d}$  and we have
$$
 \frac{N_d}{d!}\ge 
 \frac{T_d}{(Cd\log T_d)^d} \ge \frac{T_d}{(C\log^2 T_d)^d} \ge \frac{T_d}{(\log T_d)^{O(d)}},
$$
which finishes the proof.
\end{proof}

\end{document}